\documentclass[11pts, foot]{amsart}
\usepackage{amssymb}
\usepackage{amsmath}
\usepackage{thmtools}
\usepackage{thm-restate}
\usepackage{enumerate}
\usepackage{hyperref}

\usepackage{cleveref}
\usepackage[foot]{amsaddr}

\usepackage{breqn}
\usepackage{pstricks}
\usepackage{pst-tree}
\usepackage{pst-node}
\usepackage{pst-plot}
\psset{unit=0.85pt}

\begin{document}
\newcommand{\cl}[3]{\mbox{$#1,#2 \to #3$}}
\newcommand{\h}{h}
\newcommand{\p}{p}
\newcommand{\pcn}{pcn}
\newcommand{\g}{g}

\newcommand{\Heads}{\mathrm{Heads}}

\newcommand{\edge}[2]{\mbox{$(#1,#2)$}}
\newcommand{\arc}[2]{\mbox{$(#1,#2)$}}

\newtheorem{theorem}{Theorem}[section]
\newtheorem{lemma}[theorem]{Lemma}
\newtheorem{claim}{Claim}
\newtheorem{corollary}[theorem]{Corollary}
\newtheorem{proposition}[theorem]{Proposition}
\newtheorem{conjecture}[theorem]{Conjecture}
\newtheorem{problem}[theorem]{Problem}
\newtheorem{definition}[theorem]{Definition}
\newtheorem{remark}[theorem]{Remark}

\newtheorem{example}[theorem]{Example}

%\begin{frontmatter}

\title{Hydras: Directed Hypergraphs and Horn Formulas}
\author{Robert H. Sloan$^*$} \email{sloan@cs.uic.edu}\address{$^*$University of Illinois at Chicago}
\author{Despina Stasi$^{\dagger}$} \email{despina.stasi@gmail.com}\address{$^\dagger$University of Cyprus and Illinois Institute of Technology}
\author{Gy\"{o}rgy Tur\'{a}n$^{*\,\ddagger}$} \email{gyt@uic.edu}\address{$^\ddagger$MTA-SZTE Research Group on Artificial Intelligence, University of Szeged}

\begin{abstract}
We introduce a new graph parameter, the \emph{hydra number}, arising from the minimization problem for Horn formulas in propositional logic. The hydra number of a graph $G=(V,E)$ is the minimal number of hyperarcs of the form $\cl{u}{v}{w}$ required in a directed hypergraph $H=(V,F)$, such that for every pair $(u, v)$, the set of vertices reachable in $H$ from $\{u, v\}$ is the entire vertex set $V$ if $(u, v) \in E$, and it is $\{u, v\}$ otherwise. Here reachability is defined by forward chaining, a standard marking algorithm.

Various bounds are given for the hydra number. We show that the hydra number of a graph can be upper bounded by the number of edges plus the path cover number of the line graph of a spanning subgraph, which is a sharp bound in several cases. On the other hand, we construct single-headed graphs for which that bound is off by a constant factor.
Furthermore, we characterize trees with low hydra number, and give a lower bound for the hydra number of trees based on the number of vertices that are leaves in
the tree obtained from $T$ by deleting its leaves. This bound is sharp for some families of trees. We give bounds for the hydra number of complete binary trees and also discuss a related minimization problem.
\end{abstract}

\maketitle

\section{Introduction}

We consider a problem concerning the minimal number of hyperarcs in directed hypergraphs with prescribed reachability properties. In this paper a directed hypergraph $H=(V,F)$ has size-3 hyperarcs of the form $\cl{u}{v}{w}$ where $u,v$ is called the \emph{body} (or tail) and $w$ is called the \emph{head} of the hyperarc. Reachability is defined by a marking procedure known as \emph{forward chaining}. A vertex $w \in V$ is \emph{reachable} from a set $S \subseteq V$ if the following process marks $w$: start by marking vertices in $S$, and as long as there is a hyperarc $\cl{a}{b}{c}$ such that $a$ and $b$ are both marked and $c$ is unmarked, mark $c$ as well.

Given an undirected \emph{graph} $G=(V,E)$, we would like to find the minimal number of hyperarcs in a directed hypergraph $H=(V,F)$, such that for every pair $(u, v)$, the set of vertices reachable from $\{u, v\}$ in $H$ is the whole vertex set $V$ if $(u, v) \in E$, and is $\{u,v\}$ otherwise.
In other words, given a set of hypergraph bodies (specified by a graph), we look for the minimal total number of heads assigned to these bodies such that every body can reach every vertex. We call the minimum the \emph{hydra number}
\footnote{We remind the reader that in Greek mythology the Lernaean Hydra is a beast possessing many heads.}
\footnote{Portions of this work have appeared in \cite{HydrasWG2012} and the second author's doctoral dissertation~\cite{StasiThesis}.}
 of $G$, denoted by $h(G)$.

\subsection{Motivation from propositional logic}
The problem is a combinatorial reformulation of a special case of the \emph{minimization problem for propositional Horn formulas}. Horn formulas are a basic knowledge representation formalism. Horn minimization is the problem of finding a shortest possible Horn formula equivalent to a given formula. There are approximation algorithms and  computational hardness and inapproximability results for this problem \cite{AusielloDaSa1986,BhattacharyaDaMuTu2010,BorosGr2014,HammerKo1993}. Special cases correspond to the well studied transitive reduction and minimum equivalent digraph problems for directed graphs. Estimating the size of a minimal formula is not well understood even in rather simple cases. A \emph{hydra formula} $\varphi$ is a definite Horn formula with clauses of size 3 such that every body occurring in the formula occurs with all possible heads. The minimal number of clauses needed to represent $\varphi$ equals the hydra number of the undirected graph $G$ corresponding to the bodies in $\varphi$.

Besides being a natural subproblem of Horn minimization, the hydra minimization problem is also of interest for the following reason. The Horn \emph{body minimization} problem is the problem of finding, given a definite Horn formula, an equivalent Horn formula with the minimal number of distinct bodies. There are efficient algorithms for this problem \cite{AngluinFrPi1992,AriasBa2011,GuiguesDu1986,Maier1983}. Thus one possible approach to Horn minimization is to find an equivalent formula with the minimal number of bodies and then to select as few heads as possible from the set of heads assigned to the bodies. This approach
is indeed used in an approximate Horn minimization algorithm \cite{BhattacharyaDaMuTu2010}. Hydras are a natural test case for this approach.

The hydra property in Horn formulas is also of interest to relational database theory. It
 corresponds to the Boyce-Codd normal form, where every determinant set of attributes is a candidate key, that is, it either determines all attributes, or only trivial ones, i.e., attributes that are subsets of the determinant.

\subsection{Summary of results}
The hydra number of a graph is defined in Section~\ref{sect:prel}, where we also explain the Horn formula minimization problem and show, in Proposition~\ref{pro:impl},
that the two minimization problems are indeed equivalent. For the rest of the paper we use the hypergraph terminology. In Section~\ref{sect:hydras}, we begin by observing that the hydra number of a graph is at least the number of edges and at most twice the number of edges, for every graph on at least three vertices. This determines the ``playing field'' for the hydra number.
It is useful to consider the \emph{excess number} of a graph, which is the difference of its hydra number and its number of edges.
Graphs satisfying the lower bound, i.e., graphs with excess 0, are called \emph{single-headed}, and deserve further study as the best behaving graphs for the hydra number.
It turns out that single-headedness is related to, but is more general than Hamiltonicity.
 We also note that adding edges to a graph can only mildly increase its hydra number (Proposition~\ref{prop:SpanningSubgraph}) and we give some simple sufficient and some simple necessary conditions for single-headedness that we will apply later in the paper.

Theorem~\ref{thm:HydraUpperBound} in Section~\ref{sect:LineGraphs} gives a general upper bound: it shows that the hydra number is at most
the number of edges plus its \emph{generalized line graph path cover number} $p(G)$, the path cover number of the line graph of any spanning subgraph with no isolated vertices.
We give some bounds on the hydra number of disconnected graphs,
and discuss the connection to the total interval number of graphs. This connection is used to determine the maximal hydra number of a tree and the maximum excess number of a connected graph.
Section \ref{sect:LargePathCover} shows that the bound of Theorem~\ref{thm:HydraUpperBound} can be far from optimal.
In Theorem~\ref{thm:largePathCover} we construct a family of single-headed graphs with a linear number of edges and show that none of its subgraphs has a Hamiltonian line-graph and that, in fact, the quantity $p(G)$ is also linear.

We then focus our attention on the hydra number of trees in Sections~\ref{sect:Trees} and \ref{sect:specialTrees}. We show that single-headed trees are precisely the stars and that trees with excess one are precisely the non-star caterpillars (Theorem~\ref{thm:caterpillars}).
We also prove a general lower bound for trees (Theorem~\ref{thm:LowerBoundTrees}) by showing that the hydra number of a tree $T$ is at least its number of edges plus
$\left\lceil{\ell(T)}/{2}\right\rceil$, where $\ell(T)$ is the number of leaf vertices in the tree obtained from $T$ by deleting its leaves.
The statement generalizes to graphs with certain types of pendant trees. This lower bound is sharp and is attained by \emph{spider} trees (Corollary~\ref{cor:SpiderTreeHydraNumber}). In fact, spiders with an odd number of legs of length exactly two, also attain the upper bound of Theorem~\ref{thm:HydraUpperBound}. We conclude the discussion of spiders by considering the in-degree of vertices in optimal hydras representing a spider family (Theorem~\ref{thm:SpiderTreeInDegree}). We end the section by considering binary trees: we show that the hydra number of a complete binary tree is between $\frac{9}{8}\left|E\left(T\right)\right|$ and $\frac{17}{15}|E\left(T\right)|+1$ (Theorem~\ref{thm:BinaryTreeBounds}), and demonstrate that a connectivity restriction on the spanning subgraph in Theorem~\ref{thm:HydraUpperBound} can significantly worsen the upper bound on the hydra number even for trees (Theorem~\ref{thm:BinaryTreePaths}).

Finally, in Section~\ref{sect:CompleteGraphs} we consider the related problem of finding the minimal number of hyperarcs for which every $k$-tuple of vertices reaches every
vertex, and we give almost matching lower and upper bounds. 

\medskip
Ku\v{c}era in recent work~\cite{KuceraISAIM2014} studied the complexity of deciding single-headedness, and proved that this problem, and by consequence the problem of calculating the hydra number of a graph, is $\mathrm{NP}$-complete. He also introduced a subclass of trees for which the hydra number can be computed in polynomial time.

\section{The hydra number} \label{sect:prel}

In this section we discuss the two approaches for arriving at the notion of the hydra number and their relationship.

\subsection{The directed hypergraph viewpoint}

The \emph{closure} $cl_{H}(S)$ of a set of vertices $S$ with respect to directed hypergraph $H$ is the set of vertices marked by the forward chaining
procedure (described in the first paragraph of the paper) started from $S$. 

\begin{definition}
A directed 3-hypergraph $H = (V, F)$ \emph{represents} an undirected graph $G = (V, E)$ if
\begin{enumerate}[i.]
\item $\edge{u}{v} \in E$ implies $cl_H(u,v) = V$,
\item $\edge{u}{v}\not\in E$ implies $cl_H(u,v) = \{u, v\}$.
\end{enumerate}
\end{definition}

\begin{definition}
The \emph{hydra number} $\h(G)$ of an undirected graph $G = (V, E)$ is
\[ \min \{|F| \, : \, H = (V, F) \, \textrm{represents} \,\, G\}. \]
\end{definition}

These notions are illustrated by the following example.
General bounds for the hydra number of complete binary trees are given in Section~\ref{sect:specialTrees}.

\begin{example}
Consider $B_2$, the complete binary tree of depth 2. In Figure~\ref{fig:B2hydra} we see an example of a directed 3-hypergraph representing $B_2$. The bodies of the hyperarcs are the edges of $B_2$ and the heads of the hyperarcs of $H$ are pointed to by the arrows. Proposition~\ref{prop:bridge} shows that this is optimal, and thus $\h(B_2)=7$.
    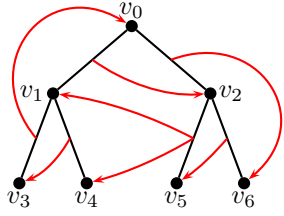
\begin{figure}[h]
    \begin{center}
    \begin{pspicture}(0,10)(100,100)
    \psset{labelsep=1.1}
    \cnode*(45,80){2.3pt}{v0}\nput{90}{v0}{$v_0$}
    \cnode*(10,50){2.3pt}{v1}\nput{180}{v1}{$v_1$}
    \cnode*(80,50){2.3pt}{v2}\nput{0}{v2}{$v_2$}
    \cnode*(-5,10){2.3pt}{v3}\nput{270}{v3}{$v_3$}
    \cnode*(25,10){2.3pt}{v4}\nput{270}{v4}{$v_4$}
    \cnode*(65,10){2.3pt}{v5}\nput{270}{v5}{$v_5$}
    \cnode*(95,10){2.3pt}{v6}\nput{270}{v6}{$v_6$}
    \ncline{v0}{v1}
    \ncline{v0}{v2}
    \ncline{v1}{v3}
    \ncline{v1}{v4}
    \ncline{v2}{v5}
    \ncline{v2}{v6}
    \pnode(27.5,65){v01}
    \pnode(62.5,65){v02}
    \pnode(2.5,30){v13}
    \pnode(17.5,30){v14}
    \pnode(72.5,30){v25}
    \pnode(87.5,30){v26}
    \ncarc[arcangle=20,linecolor=red]{->}{v14}{v3}
    \ncarc[arcangle=80,ncurv=1.3,linecolor=red]{->}{v13}{v0}
    \ncarc[arcangle=-20,linecolor=red]{->}{v01}{v2}
    \ncarc[arcangle=80,ncurv=1.3,linecolor=red]{->}{v02}{v6}
    \ncarc[linecolor=red]{->}{v26}{v5}
    \ncarc[arcangle=-10,linecolor=red]{->}{v25}{v1}
    \ncarc[linecolor=red]{->}{v25}{v4}
    \end{pspicture}
    \end{center}
    \label{fig:B2hydra}
    \caption{$H$: a directed hypergraph representing $B_2$. }
    \end{figure}
\end{example}

\begin{remark}
The removal of an isolated vertex decreases its hydra number by one.For the rest of the paper we assume that graphs contain \emph{no isolated vertices.}
\end{remark}

\subsection{The Horn formula viewpoint}

A \emph{definite Horn clause} in propositional logic is a disjunction of literals where exactly one literal is unnegated. Such a disjunction can also be viewed as an implication; for example the clause $\bar{x} \vee \bar{y} \vee z$ is equivalent to the implication $\cl{x}{y}{z}.$ The tuple $x, y$ is the \emph{body} and the variable $z$ is the \emph{head} of the clause. The size of a clause is the number of its literals. A definite $d$-Horn formula is a conjunction of definite Horn clauses of size $d$. A clause $C$ is an implicate of a formula $\varphi$ if every truth assignment
satisfying $\varphi$ satisfies $C$ as well. The implicate $C$ is a prime implicate if none of its proper subclauses is an implicate.

Implication between a definite Horn formula $\varphi$ and a definite Horn clause $C$ can be decided by the forward chaining procedure. Reformulated for propositional logic, it marks every variable in the body of $C$, and while there is a clause in $\varphi$ with all its body variables marked and its head unmarked, it marks the head of that clause as well. Then $\varphi$ implies $C$ iff the head of $C$ gets marked.

\begin{definition}
A definite 3-Horn formula $\varphi$ is a \emph{hydra} formula, or a hydra, if for every clause $\cl{x}{y}{z}$ in $\varphi$ and every variable $u$, the clause $\cl{x}{y}{u}$ also belongs to $\varphi$.
\end{definition}

For example, the formula
\[(\cl{x}{y}{z}) \wedge (\cl{x}{y}{u}) \wedge (\cl{x}{z}{y}) \wedge (\cl{x}{z}{u})\]
is a hydra formula. Redundant clauses like $\cl{x}{y}{x}$ are omitted for simplicity.

\begin{definition}
The \emph{Horn formula minimization problem for hydras} or \emph{hydra minimization problem} is the following: given a hydra formula $\varphi$, find an equivalent Horn formula with a minimal number of clauses.
\end{definition}

The following proposition is central to this study as it establishes that every \emph{prime implicate} of a hydra is a clause occurring in the hydra itself, a fact that does not hold for definite 3-Horn formulas in general. Thus minimization for hydras amounts to selecting a minimal number of clauses from the hydra that are equivalent to the original formula.

\begin{proposition} \label{pro:impl}
Every prime implicate of a hydra formula $\varphi$ is a clause of $\varphi$.
\end{proposition}
\begin{proof}
First note that all prime implicates of a definite Horn formula are definite Horn clauses~\cite{HammerKo1992}. Consider a hydra $\varphi$ and a definite Horn clause $C$. If the body of $C$ is of size 1, or it is of size 2 but it does not occur as a body in $\varphi$ then forward chaining cannot mark any further variables, thus $C$ cannot be an implicate. If the body of $C$ has size at least 3 then it
must contain a body $x, y$ occurring in $\varphi$, otherwise, again, forward chaining cannot mark any further variables. But then the clause $\cl{x}{y}{head(C)}$ occurs in $\varphi$ and so $C$ is not prime.
\end{proof}

Proposition~\ref{pro:impl} implies that \emph{the minimal formula size of a hydra $\varphi$ and the hydra number of the undirected graph $G$ formed by the bodies in $\varphi$ are the same.}

\begin{remark}
For the remainder of the paper we use hypergraph terminology.
\end{remark}

\section{First bounds on the hydra number and single-headedness}\label{sect:hydras}

In this section, starting with an easy bound for the hydra number of any graph with no isolated vertices, we define a class of graphs that provide the best behaving hydras, and give some preliminary results about their structure.

\begin{proposition}\label{prop:easyBounds}For every graph $G = (V, E)$ with at least three vertices
\[|E(G)|\leq h(G) \leq 2|E(G)|.\] \end{proposition}

\begin{proof}
For the upper bound construct a hypergraph of size $2|E(G)|$ by first ordering the edges of $G$, and then using each edge as the body of two hyperarcs whose heads are the two endpoints of the next edge in $G$. For the lower bound, note that each edge of $G$ must be a body of at least one hyperarc.
\end{proof}

Equality holds in the upper bound if and only if $G$ is a matching. Indeed, if $G$ is a matching then every edge must occur as the body of at least two hyperarcs as otherwise forward chaining cannot
mark any further vertices. If $G$ not a matching, it must contain a path of length two. Let $\edge{u}{v}$ and $\edge{v}{w}$ be the two edges of the path. Order the edges of $G$ so that $\edge{u}{v}$ and $\edge{v}{w}$ appear consecutively. We can modify the proof above noting that $\edge{u}{v}$ only needs the single head $w$. This gives a hypergraph representing $G$ with $2|E(G)|-1$ hyperarcs.

In view of the lower bound one can consider the \emph{excess} of a graph to be $h(G) - |E(G)|$. Graphs of excess 0 are of particular interest as they represent ``most compressible'' hydras.

\begin{definition}
A graph $G = (V,E)$ is \emph{single-headed} if \mbox{$h(G) = |E(G)|$.}
\end{definition}

In other words, a graph is single-headed iff there is a hypergraph $H = (V, F)$ such that every edge of $G$ has \emph{exactly} one head assigned to it, every hyperarc body in $H$ is an edge of $G$ and every pair forming an edge of $G$ has the entire vertex set as its closure. Cycles, for example, are single-headed, as shown by the directed hypergraph
\begin{equation} \label{eq:cycfor}
 (\cl{v_1}{v_2}{v_3}),(\cl{v_2}{v_3}{v_4}), \ldots, (\cl{v_{k-1}}{v_k}{v_1}).
\end{equation}

Adding edges to a cycle preserves single-headedness. For example, the graph obtained by adding edge $(v_i, v_j)$ is represented by the directed hypergraph obtained from (\ref{eq:cycfor}) by adding the hyperarc $\cl{v_i}{v_j}{v_{i+1}}$, where $i + 1$ is meant modulo $m$. Thus we obtain the following proposition. We will discuss stronger forms of this statement in the next section.

\begin{proposition} \label{prop:hamsin}
Hamiltonian graphs are single-headed.
\end{proposition}

We call a body $(u, v)$ \emph{single-headed} (resp., \emph{multi-headed}) with respect to a directed hypergraph $H$ representing a graph $G$ if it is the body of exactly one (resp., more than one) hyperarc of $H$.

\begin{remark} Assume that the directed hypergraph $H = (V, F)$ represents the graph \mbox{$G = (V, E)$} and $|V| \ge 4$. If $\cl{u}{v}{w} \in F$ and $u, v$ is single-headed in $H$ then $w$ must be a neighbor of $u$ or $v$. Indeed, otherwise $cl_{H}(u,v)=\{u,v,w\} \subset V$. This is a fact which we use numerous times in the proofs without referring to it explicitly.\end{remark}

The following proposition generalizes the argument proving Proposition~\ref{prop:hamsin}.

\begin{proposition} \label{prop:SpanningSubgraph}
Let $G$ be a graph with no isolated vertices and let $G'$ be a spanning subgraph of $G$ with no isolated vertices. Then
\[h(G)\leq h(G') + |E(G)| - |E(G')|.\]
If $G'$ is single-headed then $G$ is also single-headed.
\end{proposition}
\begin{proof}
 Let $H'$ be a directed hypergraph of size $\h(G')$ representing $G'$. Since $G'$ is a spanning subgraph of $G$ and contains no isolated vertices, for every edge $(u,v) \in E(G) \setminus E(G')$ there is an edge $(v,w)\in E(G')$. The directed hypergraph $H$ representing $G$ obtained from $H'$ by adding the hyperarc $\cl{u}{v}{w}$ to $H'$ for each edge $\edge{u}{v} \in E(G) \setminus E(G')$
satisfies the requirements. The second statement follows trivially.
\end{proof}

It may be interesting to note that if we start with a matching on $n$ vertices and keep on adding edges until the graph becomes a Hamilton cycle,
the hydra number does \emph{not} increase, each graph has hydra number $n$. We do not know if the hydra number can go down after adding an edge.

A second proposition gives a sufficient condition for single-headedness based on single-headedness of a non-spanning subgraph. This will be a crucial construction in showing that single-headed graphs can have quite complex structure in Theorem~\ref{thm:largePathCover}.
\begin{proposition} \label{prop:subgraphAddedVertex}
Let $G$ be a connected graph and $\edge{u}{v}\not\in E(G)$. Construct the graph $\hat{G}$ with vertex set $V(\hat{G})=V(G)\cup \{w\}$ and edge set $E(\hat{G})=E(G)\cup\{(u,v),(v,w)\}$, for some $w\not\in V(G)$. If $G$ is single-headed then $\hat{G}$ is single-headed.
\end{proposition}
\begin{proof}
\sloppypar{Let $H$ be a directed hypergraph representing $G$ and containing exactly $|E(G)|$ hyperarcs. Construct $\hat{H}$ from $H$ by adding hyperarcs $\cl{u}{v}{w}$ and $\cl{v}{w}{z}$, where $z$ is a neighbor of $v$ in $G$ guaranteed to exist by the connectivity of $G$. Since all pairs in $E(G)$ reach both $u$ and $v$ in $H$ (and in $\hat{H}$), hyperarc $\cl{u}{v}{w}$ ensures all pairs in $E(G)$ can reach in $\hat{H}$ the new variable $w$ as well. On the other hand, hyperarc $\cl{v}{w}{z}$ ensures that the new pairs $(u,v)$ and $(v,w)$ can reach all other variables. Finally, there are $|E(\hat{G})|$ hyperarcs in $H$.}
\end{proof}

Next we see a simple sufficient condition for a graph \emph{not} to be single-headed.

\begin{proposition}\label{prop:bridge}
Let $G$ be the union of two disjoint subgraphs \mbox{$G_1=(V_1,E_1)$} and \mbox{$G_2=(V_2,E_2)$,} connected by a cut-edge. If both $G_1$ and $G_2$ contain at least two vertices then $G$ is not single-headed.
\end{proposition}
\begin{proof}
Assume that $G$ is single-headed and let $H$ be a directed hypergraph demonstrating this. Let $u \in G_1$, $v \in G_2$, and $(u,v)$ be the cut-edge. There is exactly one hyperarc of the form $\cl{u}{v}{z}$ in $H$. If $z$ is in $G_1$ (resp., in $G_2$) then forward chaining started from $z, u$ (resp., $z, v$) cannot mark any vertices in $G_2$ (resp., $G_1$) other than $v$ (resp., $u$).
\end{proof}

\section{Line graphs and a hydra number upper bound}\label{sect:LineGraphs}

This section contains an upper bound to the hydra number of a general graph.
\subsection{Hydra number upper bound}
The \emph{line graph} $L(G)$ of $G$ has vertex set \mbox{$V(L(G))=E(G)$} and edge set $E(L(G))=\{\edge{e}{f}| e\neq f\in E(G)\textrm{ and } e\cap f \neq \emptyset\}$. A (vertex-disjoint) \emph{path cover} of $G$ is a set of vertex-disjoint paths such that every vertex $v\in V$ is in exactly one path. The \emph{path cover number} of $G$, $\pcn(G)$, is the smallest integer $k$ such that $G$ has a path cover containing $k$ paths. Also recall that a spanning subgraph $G'$ of $G$ is a subgraph with vertex set $V(G')=V(G)$ and edge set $E(G')\subseteq E(G)$. We note that a spanning subgraph need not be connected.

In Proposition~\ref{prop:hamsin} we noted that Hamiltonian graphs are single-headed. The second part of Theorem~\ref{thm:HydraUpperBound} extends this to show that Hamiltonicity of the line graph of $G$, and furthermore, Hamiltonicity of the line graph of a subgraph of $G$ with no isolated vertices, is also sufficient for single-headedness.

\begin{remark} Hamiltonicity of $L(G)$ is a strictly weaker condition than Hamiltonicity of $G$. In, particular Hamiltonicity of $G$ is easily seen to imply Hamiltonicity of the line graph, and a triangle with a pendant edge shows that the converse fails. \end{remark}

Furthermore, the number of edges plus the path cover number of the line graph of any spanning subgraph with no isolated vertices gives a general upper bound for the hydra number. Thus, for a graph $G$ we define $\p(G)$ to be the following parameter of the graph:
\[\p(G)=\min\{ \,\pcn(L(G')): \,G'\textrm{ is a spanning subgraph of $G$ with no isolated vertices}\, \}\]

\begin{theorem}\label{thm:HydraUpperBound}
Let $G$ be a graph with no isolated vertices. Then  \[\h(G)\leq \left|E(G)\right|+\p(G).\]
In addition, if there exists a spanning subgraph $G'$ with no isolated vertices such that $L(G')$ is Hamiltonian then $G$ is single-headed.
\end{theorem}

\begin{proof}
By Proposition~\ref{prop:SpanningSubgraph} it is sufficient to prove the bounds for an appropriate subgraph $G'$.

Let $k=\p(G)$, $\{P_i\}_{1}^{k}$ be the minimum path cover of $L(G')$ for some spanning subgraph $G'$ with no isolated vertices and $l_i$ be the number of vertices of the path $P_i$. Direct the edges of each path $P_i$ so that $\vec{P_i}$ is a directed path. Let $e_i=\edge{x_i}{y_i}$ and $f_i=\edge{u_i}{v_i}$ be the first and last edges in $\vec{P_i}$, respectively (if $\vec{P_i}$ is a single vertex then $e_i = f_i$).

We construct a directed hypergraph $H$ representing $G'$ and satisfying the requirements  as follows. First, for each path $\vec{P_i}$ of at least 2 vertices we add $l_{i}-1$ hyperarcs: for each directed edge $\arc{e}{f}\in\vec{P_i}$, where $e=\edge{u}{v}$ and $f=\edge{v}{w}$, add a hyperarc $\cl{u}{v}{w}$ to $H$.

If $k = 1$ then we complete the construction of $H$ by adding two hyperarcs, $\cl{u_1}{v_1}{x_1}$ and $\cl{u_1}{v_1}{y_1}$. If $k>1$ then we complete the construction by adding the $2k$ hyperarcs
\[
 (\cl{u_k}{v_k}{x_1}), (\cl{u_k}{v_k}{y_1}) \mbox{ and } (\cl{u_i}{v_i}{x_{i+1}}), (\cl{u_i}{v_i}{y_{i+1}}),
\]
 for $1\leq i\leq k-1.$ It follows directly from the construction that the hypergraph $H$  represents $G'$.

For the second part of the theorem, suppose $G'$ is a spanning subgraph with no isolated vertices such that $L(G')$ is Hamiltonian.
Let $C$ be a Hamiltonian cycle in $L(G')$. Direct the edges of $C$ so that $\vec{C}$ is a directed Hamiltonian cycle. The directed hypergraph $H$ satisfying the requirements is constructed by adding a hyperarc $\cl{u}{v}{w}$ for each directed edge $\arc{e}{f}\in\vec{C}$, where $e=\edge{u}{v}$ and $f=\edge{v}{w}$.
\end{proof}

\subsection{Disconnected graphs}

In this subsection we add a few remarks on the hydra numbers of disconnected graphs. First note that the general upper bound of $2|E(G)|$, which applies to
all graphs, can be improved for connected graphs as follows: build a hydra for a spanning tree and then add additional single-headed edges using
Proposition~\ref{prop:SpanningSubgraph}. Theorem~\ref{thm:maxHydraNumberOfTrees} gives an essentially $5/4 |V(G)|$ for the hydra number of
trees, so this way one gets an upper bound of $|E(G)| + |V(G)|/4$  for connected graphs. This can be extended to disconnected graphs by
adding two hyperarcs per connected component, pointing to the next connected component in a circular order.

On the other hand, the following lower bound is easy to see for any graph $G$ containing at least two connected components.
\begin{proposition}\label{conj:components}
Let $G$ consist of $k\geq2$ connected components $G_1, G_2, \ldots, G_k$ for \mbox{$k\geq2$}, and let each $G_i$ contain at least two vertices. Then \mbox{$\h(G)\geq |E(G)|+k.$}
\end{proposition}
\begin{proof}
Consider a directed hypergraph $H$ representing $G$. For each component $G_i$ there must be an edge $\edge{u_i}{v_i}\in G_i$ such that $\cl{u_i}{v_i}{w}\in H$ for some $w\not\in G_i$, and, since $w$ is not adjacent to $u_i$ or $v_i$ in $G$, $\edge{u_i}{v_i}$ cannot be single-headed in $H$.
\end{proof}

The relationship between the hydra number of a disconnected graph and the sum of the hydra numbers of its components is open. For single-headed components at least one multi-headed edge is needed, otherwise other components cannot be reached. Thus one can ask the following.

\begin{problem}
Let $G$ consist of $k$ connected components $G_1, G_2, \ldots, G_k$ for \mbox{$k\geq2$}, such that each $G_i$ contain at least two vertices and $s$ of the
components are single-headed. Then does the following hold:
\[ \h(G)= \sum_{i=1}^{k} h(G_i)+s.\]
\end{problem}

 It can be shown that the answer is positive when each $G_i$ is single-headed or contains a spanning caterpillar tree. This is because $\sum_{i=1}^{k} h(G_i)+s=|E(G)|+k$ by Proposition~\ref{prop:SpanningSubgraph} and Theorem~\ref{thm:caterpillars}. 

%%%%%%%%%%%%%%%
\subsection{Hydras and the total interval number of a graph}
Theorem~\ref{thm:HydraUpperBound}, in relating the hydra number of a graph $G$ with the path cover number of the line graph of its subgraphs, indicates a connection with another graph parameter, the total interval number of a graph, introduced independently by \cite{GriggsWest1980} and \cite{TrotterHarary1979}. The \emph{total interval number} of a graph $G=(V,E)$ is the minimum number of disjoint real intervals necessary in a \emph{multiple-interval representation} of a graph; the latter is an assignment multiple disjoint real intervals to each vertex of $G$ with the property that two vertices are adjacent in $G$ if and only if their assigned sets intersect.

 In \cite{KratzkeWe1996}, it was shown that for a triangle-free graph $G$, the total interval number of $G$, $\tau(G)=|E(G)|+\pcn(L(G)).$ The same result follows from \cite{Raychaudhuri1995}. Since clearly $\p(G)\leq\pcn(L(G))$ and by Theorem~{\ref{thm:HydraUpperBound}, the hydra number of $G$ is at most the total interval number of the graph establishing the following corollary.

\begin{corollary} \label{cor:TotalIntervalNumber}
For a triangle-free graph $G$ with no isolated vertices \mbox{$\h(G)\leq\tau(G).$} In addition, for a graph $G$ containing a triangle-free subgraph $G'$ with no isolated vertices, \mbox{$\h(G)\leq\tau(G').$}
\end{corollary}

In Section~\ref{subsect:Spiders} we use this corollary to bound the maximum hydra number of a graph.

\section{Single-headed graphs with large $p(G)$}\label{sect:LargePathCover}

In this section we construct a family of sparse single-headed graphs with large generalized line graph path cover number. This shows that the upper bound of Theorem~\ref{thm:HydraUpperBound} can be off by a constant factor, and that the condition of the second half of the theorem is not necessary for single-headedness.
%%%%%%%%%%%%%%%

\begin{theorem} \label{thm:largePathCover}
There is a family of single-headed graphs $\{G_k, k \geq 2\}$, where $G_k$ has $\Theta(k)$ edges and $\p(G_k)= \Theta\left(k\right)$.
\end{theorem}

\begin{proof}
\sloppypar{Consider the sequence of graphs $G_k, k\geq 2$ constructed as follows. Starting from an $8k$-cycle, with vertices $v_0,\ldots,v_{8k-1}$, and pendant edges $x_{i}v_{4i}$ and $y_iv_{4k+4i}$ for \mbox{$0\leq i\leq k-1$}, add a vertex $z_i$ and the edges $(x_i, y_i)$, $(y_i,z_i)$, for each $i$ such that $0\leq i\leq k-1$. Recall the construction in Proposition~\ref{prop:subgraphAddedVertex}.}

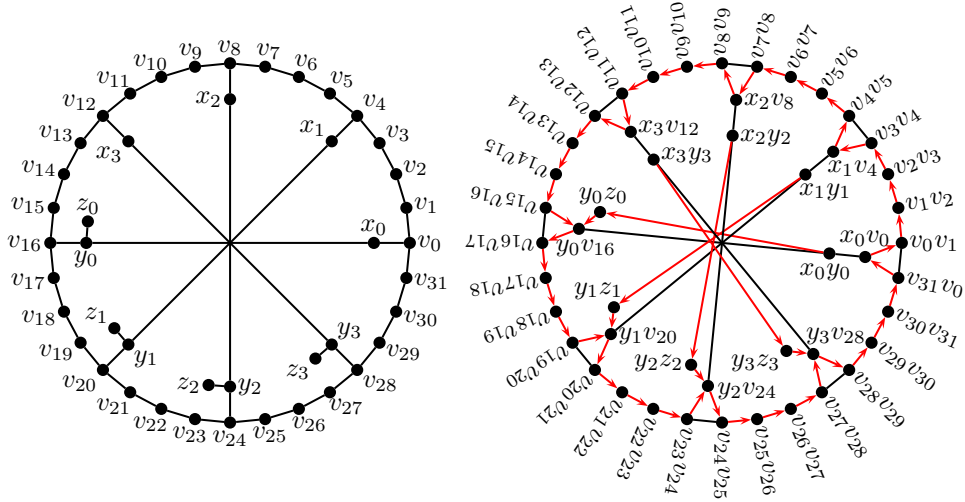
\begin{figure}[htb]
\centering
\begin{tabular}{cc}
\begin{pspicture}(30,-10)(225,210)
\psset{linewidth=.9, labelsep=0.8}
\cnode*(180.0,100.0){2.3pt}{v0}\nput{0}{v0}{$v_{0}$}
\cnode*(164.0,100.0){2.3pt}{x0}\nput{90}{x0}{$x_{0}$}
\ncline{v0}{x0}
\cnode*(178.46282243225843,115.60722576129027){2.3pt}{v1}\nput{11}{v1}{$v_{1}$}
\cnode*(173.91036260090294,130.6146745892072){2.3pt}{v2}\nput{23}{v2}{$v_{2}$}
\cnode*(166.51756898420362,144.44561864156816){2.3pt}{v3}\nput{34}{v3}{$v_{3}$}
\cnode*(156.5685424949238,156.5685424949238){2.3pt}{v4}\nput{45}{v4}{$v_{4}$}
\cnode*(145.25483399593904,145.25483399593904){2.3pt}{x1}\nput{135}{x1}{$x_{1}$}
\ncline{v4}{x1}
\cnode*(144.4456186415682,166.51756898420362){2.3pt}{v5}\nput{56}{v5}{$v_{5}$}
\cnode*(130.6146745892072,173.91036260090294){2.3pt}{v6}\nput{68}{v6}{$v_{6}$}
\cnode*(115.60722576129027,178.46282243225843){2.3pt}{v7}\nput{79}{v7}{$v_{7}$}
\cnode*(100.0,180.0){2.3pt}{v8}\nput{90}{v8}{$v_{8}$}
\cnode*(100.0,164.0){2.3pt}{x2}\nput{180}{x2}{$x_{2}$}
\ncline{v8}{x2}
\cnode*(84.39277423870975,178.46282243225843){2.3pt}{v9}\nput{101}{v9}{$v_{9}$}
\cnode*(69.38532541079283,173.91036260090294){2.3pt}{v10}\nput{113}{v10}{$v_{10}$}
\cnode*(55.55438135843184,166.51756898420362){2.3pt}{v11}\nput{124}{v11}{$v_{11}$}
\cnode*(43.4314575050762,156.5685424949238){2.3pt}{v12}\nput{135}{v12}{$v_{12}$}
\cnode*(54.74516600406096,145.25483399593904){2.3pt}{x3}\nput{225}{x3}{$x_{3}$}
\ncline{v12}{x3}
\cnode*(33.48243101579638,144.44561864156816){2.3pt}{v13}\nput{146}{v13}{$v_{13}$}
\cnode*(26.089637399097057,130.6146745892072){2.3pt}{v14}\nput{158}{v14}{$v_{14}$}
\cnode*(21.537177567741566,115.6072257612903){2.3pt}{v15}\nput{169}{v15}{$v_{15}$}
\cnode*(20.0,100.00000000000001){2.3pt}{v16}\nput{180}{v16}{$v_{16}$}
\cnode*(36.0,100.00000000000001){2.3pt}{y0}\nput{270}{y0}{$y_{0}$}
\ncline{v16}{y0}
\cnode*(36.7186510120933,109.56404047831035){2.3pt}{z0}\nput{90}{z0}{$z_{0}$}
\ncline{y0}{z0}
\cnode*(21.537177567741566,84.39277423870973){2.3pt}{v17}\nput{191}{v17}{$v_{17}$}
\cnode*(26.089637399097057,69.38532541079283){2.3pt}{v18}\nput{202}{v18}{$v_{18}$}
\cnode*(33.482431015796365,55.55438135843184){2.3pt}{v19}\nput{214}{v19}{$v_{19}$}
\cnode*(43.43145750507618,43.4314575050762){2.3pt}{v20}\nput{225}{v20}{$v_{20}$}
\cnode*(54.74516600406095,54.74516600406096){2.3pt}{y1}\nput{315}{y1}{$y_{1}$}
\ncline{v20}{y1}
\cnode*(48.490531130262816,62.01612688577458){2.3pt}{z1}\nput{135}{z1}{$z_{1}$}
\ncline{y1}{z1}
\cnode*(55.554381358431826,33.48243101579638){2.3pt}{v21}\nput{236}{v21}{$v_{21}$}
\cnode*(69.38532541079277,26.089637399097086){2.3pt}{v22}\nput{247}{v22}{$v_{22}$}
\cnode*(84.3927742387097,21.53717756774158){2.3pt}{v23}\nput{259}{v23}{$v_{23}$}
\cnode*(99.99999999999999,20.0){2.3pt}{v24}\nput{270}{v24}{$v_{24}$}
\cnode*(99.99999999999999,36.0){2.3pt}{y2}\nput{360}{y2}{$y_{2}$}
\ncline{v24}{y2}
\cnode*(90.43595952168961,36.7186510120933){2.3pt}{z2}\nput{180}{z2}{$z_{2}$}
\ncline{y2}{z2}
\cnode*(115.60722576129027,21.537177567741566){2.3pt}{v25}\nput{281}{v25}{$v_{25}$}
\cnode*(130.6146745892072,26.08963739909707){2.3pt}{v26}\nput{293}{v26}{$v_{26}$}
\cnode*(144.44561864156816,33.482431015796365){2.3pt}{v27}\nput{304}{v27}{$v_{27}$}
\cnode*(156.5685424949238,43.43145750507618){2.3pt}{v28}\nput{315}{v28}{$v_{28}$}
\cnode*(145.25483399593904,54.74516600406095){2.3pt}{y3}\nput{405}{y3}{$y_{3}$}
\ncline{v28}{y3}
\cnode*(137.9838731142254,48.4905311302628){2.3pt}{z3}\nput{225}{z3}{$z_{3}$}
\ncline{y3}{z3}
\cnode*(166.51756898420362,55.554381358431826){2.3pt}{v29}\nput{326}{v29}{$v_{29}$}
\cnode*(173.9103626009029,69.38532541079277){2.3pt}{v30}\nput{337}{v30}{$v_{30}$}
\cnode*(178.4628224322584,84.3927742387097){2.3pt}{v31}\nput{349}{v31}{$v_{31}$}
\ncline{v0}{v1}
\ncline{v1}{v2}
\ncline{v2}{v3}
\ncline{v3}{v4}
\ncline{v4}{v5}
\ncline{v5}{v6}
\ncline{v6}{v7}
\ncline{v7}{v8}
\ncline{v8}{v9}
\ncline{v9}{v10}
\ncline{v10}{v11}
\ncline{v11}{v12}
\ncline{v12}{v13}
\ncline{v13}{v14}
\ncline{v14}{v15}
\ncline{v15}{v16}
\ncline{v16}{v17}
\ncline{v17}{v18}
\ncline{v18}{v19}
\ncline{v19}{v20}
\ncline{v20}{v21}
\ncline{v21}{v22}
\ncline{v22}{v23}
\ncline{v23}{v24}
\ncline{v24}{v25}
\ncline{v25}{v26}
\ncline{v26}{v27}
\ncline{v27}{v28}
\ncline{v28}{v29}
\ncline{v29}{v30}
\ncline{v30}{v31}
\ncline{v31}{v0}
\ncline{x0}{y0}
\ncline{x1}{y1}
\ncline{x2}{y2}
\ncline{x3}{y3}
\end{pspicture}
\begin{pspicture}(10,-10)(170,190)
\psset{linewidth=.9, labelsep=0.8}
\cnode*(180.0,100.0){2.3pt}{v0v1}\nput[rot=0]{0}{v0v1}{$v_{0}v_{1}$}
\cnode*(178.46282243225843,115.60722576129027){2.3pt}{v1v2}\nput[rot=11]{11}{v1v2}{$v_{1}v_{2}$}
\cnode*(173.91036260090294,130.6146745892072){2.3pt}{v2v3}\nput[rot=23]{23}{v2v3}{$v_{2}v_{3}$}
\cnode*(166.51756898420362,144.44561864156816){2.3pt}{v3v4}\nput[rot=34]{34}{v3v4}{$v_{3}v_{4}$}
\cnode*(156.5685424949238,156.5685424949238){2.3pt}{v4v5}\nput[rot=45]{45}{v4v5}{$v_{4}v_{5}$}
\cnode*(144.4456186415682,166.51756898420362){2.3pt}{v5v6}\nput[rot=56]{56}{v5v6}{$v_{5}v_{6}$}
\cnode*(130.6146745892072,173.91036260090294){2.3pt}{v6v7}\nput[rot=68]{68}{v6v7}{$v_{6}v_{7}$}
\cnode*(115.60722576129027,178.46282243225843){2.3pt}{v7v8}\nput[rot=79]{79}{v7v8}{$v_{7}v_{8}$}
\cnode*(100.0,180.0){2.3pt}{v8v9}\nput[rot=90]{90}{v8v9}{$v_{8}v_{9}$}
\cnode*(84.39277423870975,178.46282243225843){2.3pt}{v9v10}\nput[rot=101]{101}{v9v10}{$v_{9}v_{10}$}
\cnode*(69.38532541079283,173.91036260090294){2.3pt}{v10v11}\nput[rot=113]{113}{v10v11}{$v_{10}v_{11}$}
\cnode*(55.55438135843184,166.51756898420362){2.3pt}{v11v12}\nput[rot=124]{124}{v11v12}{$v_{11}v_{12}$}
\cnode*(43.4314575050762,156.5685424949238){2.3pt}{v12v13}\nput[rot=135]{135}{v12v13}{$v_{12}v_{13}$}
\cnode*(33.48243101579638,144.44561864156816){2.3pt}{v13v14}\nput[rot=146]{146}{v13v14}{$v_{13}v_{14}$}
\cnode*(26.089637399097057,130.6146745892072){2.3pt}{v14v15}\nput[rot=158]{158}{v14v15}{$v_{14}v_{15}$}
\cnode*(21.537177567741566,115.6072257612903){2.3pt}{v15v16}\nput[rot=169]{169}{v15v16}{$v_{15}v_{16}$}
\cnode*(20.0,100.00000000000001){2.3pt}{v16v17}\nput[rot=180]{180}{v16v17}{$v_{16}v_{17}$}
\cnode*(21.537177567741566,84.39277423870973){2.3pt}{v17v18}\nput[rot=191]{191}{v17v18}{$v_{17}v_{18}$}
\cnode*(26.089637399097057,69.38532541079283){2.3pt}{v18v19}\nput[rot=202]{202}{v18v19}{$v_{18}v_{19}$}
\cnode*(33.482431015796365,55.55438135843184){2.3pt}{v19v20}\nput[rot=214]{214}{v19v20}{$v_{19}v_{20}$}
\cnode*(43.43145750507618,43.4314575050762){2.3pt}{v20v21}\nput[rot=225]{225}{v20v21}{$v_{20}v_{21}$}
\cnode*(55.554381358431826,33.48243101579638){2.3pt}{v21v22}\nput[rot=236]{236}{v21v22}{$v_{21}v_{22}$}
\cnode*(69.38532541079277,26.089637399097086){2.3pt}{v22v23}\nput[rot=247]{247}{v22v23}{$v_{22}v_{23}$}
\cnode*(84.3927742387097,21.53717756774158){2.3pt}{v23v24}\nput[rot=259]{259}{v23v24}{$v_{23}v_{24}$}
\cnode*(99.99999999999999,20.0){2.3pt}{v24v25}\nput[rot=270]{270}{v24v25}{$v_{24}v_{25}$}
\cnode*(115.60722576129027,21.537177567741566){2.3pt}{v25v26}\nput[rot=281]{281}{v25v26}{$v_{25}v_{26}$}
\cnode*(130.6146745892072,26.08963739909707){2.3pt}{v26v27}\nput[rot=293]{293}{v26v27}{$v_{26}v_{27}$}
\cnode*(144.44561864156816,33.482431015796365){2.3pt}{v27v28}\nput[rot=304]{304}{v27v28}{$v_{27}v_{28}$}
\cnode*(156.5685424949238,43.43145750507618){2.3pt}{v28v29}\nput[rot=315]{315}{v28v29}{$v_{28}v_{29}$}
\cnode*(166.51756898420362,55.554381358431826){2.3pt}{v29v30}\nput[rot=326]{326}{v29v30}{$v_{29}v_{30}$}
\cnode*(173.9103626009029,69.38532541079277){2.3pt}{v30v31}\nput[rot=337]{337}{v30v31}{$v_{30}v_{31}$}
\cnode*(178.4628224322584,84.3927742387097){2.3pt}{v31v0}\nput[rot=349]{349}{v31v0}{$v_{31}v_{0}$}
\ncline[linecolor=red]{->}{v0v1}{v1v2}
\ncline[linecolor=red]{->}{v1v2}{v2v3}
\ncline[linecolor=red]{->}{v2v3}{v3v4}
\ncline{v3v4}{v4v5}
\ncline[linecolor=red]{->}{v4v5}{v5v6}
\ncline[linecolor=red]{->}{v5v6}{v6v7}
\ncline[linecolor=red]{->}{v6v7}{v7v8}
\ncline{v7v8}{v8v9}
\ncline[linecolor=red]{->}{v8v9}{v9v10}
\ncline[linecolor=red]{->}{v9v10}{v10v11}
\ncline[linecolor=red]{->}{v10v11}{v11v12}
\ncline{v11v12}{v12v13}
\ncline[linecolor=red]{->}{v12v13}{v13v14}
\ncline[linecolor=red]{->}{v13v14}{v14v15}
\ncline[linecolor=red]{->}{v14v15}{v15v16}
\ncline{v15v16}{v16v17}
\ncline[linecolor=red]{->}{v16v17}{v17v18}
\ncline[linecolor=red]{->}{v17v18}{v18v19}
\ncline[linecolor=red]{->}{v18v19}{v19v20}
\ncline{v19v20}{v20v21}
\ncline[linecolor=red]{->}{v20v21}{v21v22}
\ncline[linecolor=red]{->}{v21v22}{v22v23}
\ncline[linecolor=red]{->}{v22v23}{v23v24}
\ncline{v23v24}{v24v25}
\ncline[linecolor=red]{->}{v24v25}{v25v26}
\ncline[linecolor=red]{->}{v25v26}{v26v27}
\ncline[linecolor=red]{->}{v26v27}{v27v28}
\ncline{v27v28}{v28v29}
\ncline[linecolor=red]{->}{v28v29}{v29v30}
\ncline[linecolor=red]{->}{v29v30}{v30v31}
\ncline[linecolor=red]{->}{v30v31}{v31v0}
\ncline{v31v0}{v0v1}
\cnode*(163.69182250702062,93.72690301890812){2.3pt}{x0v0}%\nput{129}{x0v0}{$x_{0}v_{0}$}
\nput[labelsep=2pt]{90}{x0v0}{$x_{0}v_{0}$}
\ncline[linecolor=red]{->}{x0v0}{v0v1}
\ncline[linecolor=red]{->}{v31v0}{x0v0}
\cnode*(149.47266901521516,140.60117018647333){2.3pt}{x1v4}%\nput{174}{x1v4}{$x_{1}v_{4}$}
\nput[labelsep=1.5pt]{-60}{x1v4}{$x_{1}v_{4}$}
\ncline[linecolor=red]{->}{x1v4}{v4v5}
\ncline[linecolor=red]{->}{v3v4}{x1v4}
\cnode*(106.2730969810919,163.6918225070206){2.3pt}{x2v8}%\nput{219}{x2v8}{$x_{2}v_{8}$}
\nput{-6}{x2v8}{$x_{2}v_{8}$}
\ncline[linecolor=red]{->}{x2v8}{v8v9}
\ncline[linecolor=red]{->}{v7v8}{x2v8}
\cnode*(59.398829813526696,149.47266901521516){2.3pt}{x3v12}%\nput{264}{x3v12}{$x_{3}v_{12}$}
\nput{15}{x3v12}{$x_{3}v_{12}$}
\ncline[linecolor=red]{->}{x3v12}{v12v13}
\ncline[linecolor=red]{->}{v11v12}{x3v12}
\cnode*(36.308177492979404,106.2730969810919){2.3pt}{y0v16}%\nput{309}{y0v16}{$y_{0}v_{16}$}
\nput[labelsep=2pt]{-80}{y0v16}{$y_{0}v_{16}$}
\ncline[linecolor=red]{->}{y0v16}{v16v17}
\ncline[linecolor=red]{->}{v15v16}{y0v16}
\cnode*(50.527330984784825,59.3988298135267){2.3pt}{y1v20}\nput{354}{y1v20}{$y_{1}v_{20}$}
\ncline[linecolor=red]{->}{y1v20}{v20v21}
\ncline[linecolor=red]{->}{v19v20}{y1v20}
\cnode*(93.72690301890813,36.3081774929794){2.3pt}{y2v24}%\nput{399}{y2v24}{$y_{2}v_{24}$}
\nput[labelsep=1.5pt]{-15}{y2v24}{$y_{2}v_{24}$}
\ncline[linecolor=red]{->}{y2v24}{v24v25}
\ncline[linecolor=red]{->}{v23v24}{y2v24}
\cnode*(140.60117018647333,50.52733098478484){2.3pt}{y3v28}%\nput{444}{y3v28}{$y_{3}v_{28}$}
\nput[labelsep=2pt]{60}{y3v28}{$y_{3}v_{28}$}
\ncline[linecolor=red]{->}{y3v28}{v28v29}
\ncline[linecolor=red]{->}{v27v28}{y3v28}
\cnode*(147.76886688026545,95.29517726418109){2.3pt}{x0y0}%\nput{-96}{x0y0}{$x_{0}y_{0}$}
\nput[labelsep=1.5pt]{-96}{x0y0}{$x_{0}y_{0}$}
\ncline{x0v0}{x0y0}
\ncline{x0y0}{y0z0}
\cnode*(137.10450176141137,130.45087763985498){2.3pt}{x1y1}%\nput{-51}{x1y1}{$x_{1}y_{1}$}
\nput[labelsep=1.5pt]{-60}{x1y1}{$x_{1}y_{1}$}
\ncline{x1v4}{x1y1}
\ncline{x1y1}{y1z1}
\cnode*(104.70482273581892,147.76886688026545){2.3pt}{x2y2}\nput{-6}{x2y2}{$x_{2}y_{2}$}
\ncline{x2v8}{x2y2}
\ncline{x2y2}{y2z2}
\cnode*(69.54912236014502,137.10450176141137){2.3pt}{x3y3}%\nput{39}{x3y3}{$x_{3}y_{3}$}
\nput{15}{x3y3}{$x_{3}y_{3}$}
\ncline{x3v12}{x3y3}
\ncline{x3y3}{y3z3}
\cnode*(45.71570691398952,113.7555633891206){2.3pt}{y0z0}\nput{84}{y0z0}{$y_{0}z_{0}$}
\ncline[linecolor=red]{<-}{y0v16}{y0z0}
\ncline{y0v16}{x0y0}
\ncline[linecolor=red]{<-}{y0z0}{x0y0}
\cnode*(51.888556095475394,71.34186039845255){2.3pt}{y1z1}%\nput{129}{y1z1}{$y_{1}z_{1}$}
\nput[labelsep=1.5]{115}{y1z1}{$y_{1}z_{1}$}
\ncline[linecolor=red]{<-}{y1v20}{y1z1}
\ncline{y1v20}{x1y1}
\ncline[linecolor=red]{<-}{y1z1}{x1y1}
\cnode*(86.2444366108794,45.71570691398952){2.3pt}{y2z2}\nput{174}{y2z2}{$y_{2}z_{2}$}
\ncline[linecolor=red]{<-}{y2v24}{y2z2}
\ncline{y2v24}{x2y2}
\ncline[linecolor=red]{<-}{y2z2}{x2y2}
\cnode*(128.65813960154745,51.888556095475394){2.3pt}{y3z3}\nput{219}{y3z3}{$y_{3}z_{3}$}
\ncline[linecolor=red]{<-}{y3v28}{y3z3}
\ncline{y3v28}{x3y3}
\ncline[linecolor=red]{<-}{y3z3}{x3y3}
\end{pspicture}
\end{tabular}
\caption{Single-headed graph $G_k$ for $k=4$ from Theorem~\ref{thm:largePathCover} (left) and its line graph (right). The arcs on the line graph indicate of a smallest-possible directed hypergraph realizing the hydra. }
\label{fig:largePathCover}
\end{figure}

By Proposition~\ref{prop:subgraphAddedVertex}, $G_k$ is single-headed, since a cycle with attached pendant edges has a Hamiltonian line graph. We will show that for an arbitrary, not necessarily connected, spanning subgraph $G'\subseteq G_k$ the path cover number of $L(G')$ is at least $k/4$.

Define $D_i$ to be the set of vertices in the $i_{\textrm{th}}$ diagonal of $L(G)$, namely $x_{i}v_{4i}$, $x_{i}y_{i}$, $y_{i}z_{i}$, and $y_{i}v_{4k+4i}$. Consider an arbitrary path cover \mbox{$S=\{P_j: 1\leq j \leq s\}$} of the vertices of $L(G')$.

\begin{lemma}\label{lem:diagonal}
Let ${D_i}'=D_i\cap V(L(G'))$, and let $G[{D_i}']$ be the subgraph of $L(G')$ induced by ${D_i}'$. If $G[{D_i}']$ does not contain an endpoint of a path in $S$, then ${D_i}'=D_i$ and one path in $S$ covers all vertices in $D_i$.
\end{lemma}
\begin{proof}
Suppose that $G[{D_i}']$ does not contain an endpoint of a path in $S$, and assume for contradiction that $e\in D_i\setminus {D_i}'$.
Since $G'$ is spanning and contains no isolated vertices, it must contain the edge $(y_i,z_i)$, and so $e\neq y_iz_i$. Also $e\not\in\{ x_iy_i, y_iv_{4k+4i}\}$, else $y_iz_i$ would be a degree-1 vertex in $L(G')$, and thus it would be an endpoint of a path in $S$. Furthermore $e\neq x_iv_{4i}$, otherwise $S$ must have a path endpoint in the triangle $\{x_iy_i, y_iz_i, y_iv_{4k+4i}\}$. Thus $D_i$ is contained in the vertex set of $L(G')$, and due to the structure of the diagonal and the assumption that no path endpoints of $S$ fall in $D_i$, all vertices in the diagonal are covered by exactly one path $P$ of $S$.
\end{proof}

Define $X_i$ to include all vertices in $D_i$ along with the cycle vertices $v_{4i-3}v_{4i-2}$, $v_{4i-2}v_{4i-1}$, $v_{4i-1}v_{4i}$, $v_{4i}v_{4i+1}$, $v_{4i+1}v_{4i+2}$, $v_{4i+2}v_{4i+3}$, and their antipodes on the circle $v_{4k+4i-3}v_{4k+4i-2}$, $v_{4k+4i-2}v_{4k+4i-1}$, $v_{4k+4i-1}v_{4k+4i}$, $v_{4k+4i}v_{4k+4i+1}$, $v_{4k+4i+1}v_{4k+4i+2}$, $v_{4k+4i+2}v_{4k+4i+3}$.

Let ${X_i}'=X_i\cap V(L(G'))$. We claim that the subgraph $G[{X_i}']$ induced by the vertex set ${X_i}'$ contains at least one endpoint of a path in $S$. Suppose not. By Lemma~\ref{lem:diagonal} all vertices in $D_i$ are in $L(G')$. A case analysis shows that all other vertices in $X_i$ must be present, otherwise a degree-1, or degree-0 vertex is introduced in $G[{X_i}']$ or $G'$ is not both spanning and without isolated vertices.
Indeed, deleting two consecutive vertices along the cycle $v_{i-1}v_i$  and $v_iv_{i+1}$ would isolate vertex $v_i$ in $G'$, and deleting any one vertex $v_iv_{i+1}$, or any two non-consecutive vertices, along the cycle would make at least one of $v_{i-1}v_i$ and $v_{i+1}v_{i+2}$ a degree-0 or degree-1 vertex in $G[{X_i}']$. Thus there must be a path $P$ in $S$ going through all the vertices of $X_i$, and by hypothesis it has no endpoints in $X_i$.  But  this is not possible as no such path can include all three of the vertices $v_{4i-1}v_{4i}$, $x_iy_i$, and $v_{4i+2}v_{4i+3}$.

There are $k/2$ disjoint sets ${X_i}'$ and so there are at least $k/4$ paths in $S$.
\end{proof}

%%%%%%%%%%%%%%%%%%%%%%%%%%%%%%%%
%%%%%%%%%%%%%%%%%%%%%%%%%%%%%%%%

\section{The hydra number of trees}\label{sect:Trees}

In this section we obtain bounds for the hydra number in the special case of trees. This is a natural subclass to consider. Furthermore,
Proposition~\ref{prop:SpanningSubgraph} shows that the excess of a graph does not increase by adding edges and thus trees actually give examples
of graphs with maximal excess.

\subsection{Trees with Low Hydra Number}\label{subsect:TreesLowHydraNumber}
We begin the discussion of the hydra number of trees, with trees having excess 0 or 1.

A \emph{star} is a tree that contains no length-3 path. A \emph{caterpillar} is a tree for which deleting all vertices of degree one and their incident edges from the tree gives a path. We call this path the spine of $T$, and note that it is unique. A useful characterization of caterpillars is that they do not contain the subgraph in Figure~\ref{fig:caterpillar} \cite{HarariSc1971} (see also \cite[p.88]{West2001}).
\begin{figure}[hbt]
\centering
\begin{pspicture}(20,25)(50,75)
\psset{linewidth=.9, labelsep=0.8, unit=1pt}
\cnode*(30,60){2.3pt}{v0}%\nput{180}{v0}{$x$}%\nput{0}{v0}{$v_{0}$}
\cnode*(30,50){2.3pt}{v1}%\nput{180}{v1}{$y$}%\nput{0}{v1}{$v_{1}$}
\cnode*(30,40){2.3pt}{v2}%\nput{0}{v2}{$u$}%\nput{0}{v2}{$v_{2}$}
\cnode*(25,30){2.3pt}{v3}%\nput{180}{v3}{$v$}%\nput{180}{v3}{$v_{3}$}
\cnode*(35,30){2.3pt}{v4}%\nput{0}{v4}{$s$}%\nput{0}{v4}{$v_{4}$}
\cnode*(20,20){2.3pt}{v5}%\nput{180}{v5}{$w$}%\nput{180}{v5}{$v_{5}$}
\cnode*(40,20){2.3pt}{v6}%\nput{0}{v6}{$t$}%\nput{0}{v6}{$v_{6}$}
\ncline{v0}{v1}
\ncline{v1}{v2}
\ncline{v2}{v3}
\ncline{v3}{v5}
\ncline{v2}{v4}
\ncline{v4}{v6}
\end{pspicture}
\caption{The forbidden subgraph for caterpillars.}
\label{fig:caterpillar}
\end{figure}
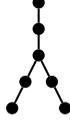

Caterpillars have been instrumental in \cite{Raychaudhuri1995}, where finding maximal caterpillars starting from the leaves of the tree was the basis for a polynomial algorithm used to find a minimum Hamiltonian completion of the line graph of a tree
(which is the same as finding a minimum path cover). A linear algorithm was later put forth by \cite{AgnetisDeMePa2001} for the same problem. For general graphs the problem is NP-hard. Furthermore, \cite{Bertossi1981} proves that finding a Hamiltonian path is NP-complete even for line graphs.

Stars are the only trees that are single-headed, and caterpillars are the only non-star trees that can attain $\h(T)=|E(T)|+1$.
\begin{theorem}\label{thm:caterpillars}
Let $T$ be a tree. Then
\begin{enumerate}[i.]
\item \mbox{$\h(T)=|E(T)|$} if and only if $T$ is a star.
\item \mbox{$\h(T)=|E(T)|+1$} if and only if $T$ is a non-star caterpillar.
\end{enumerate}
\end{theorem}
\begin{proof}
First we prove the upper bounds. The single-headedness of stars is easily seen directly, or follows from Theorem~\ref{thm:HydraUpperBound}. For caterpillars, the upper bound follows from Theorem~\ref{thm:HydraUpperBound} as the line graph of a caterpillar contains a Hamiltonian path.

For the lower bounds, note that if a tree is not a star then it contains a path of length three. The middle edge of the path is a cut-edge between two components of at least two vertices, hence Proposition~\ref{prop:bridge} implies that the tree is not single-headed. The lower bound for non-star caterpillars follows from Theorem~\ref{thm:LowerBoundTrees} in the next subsection.
\end{proof}

\subsection{Lower Bounds for Trees}\label{subsect:LowerBoundTrees}

Let $T^-$ be the tree obtained from $T$ by removing all leaves, and define $\ell(T)$ to be the number of leaves in $T^-$. 
\begin{theorem}\label{thm:LowerBoundTrees}
For a tree $T$ that is not a star,
    \[h(T)\geq|E(T)|+\left\lceil\ell(T)/2\right\rceil.\]
\end{theorem}
\begin{proof}

If $T$ is the tree consisting of a central edge $\edge{u}{v}$ and pendant edges attached to both $u$ and $v$, then $\edge{u}{v}$ is cut edge. By Proposition~\ref{prop:bridge}, $T$ is not single-headed and the bound $\h(T)\geq |E(T)|+\left\lceil\ell(T)/2\right\rceil=|E(T)|+1$ holds.

So we may assume that $T$ is not that tree. Suppose that $H$ is a hypergraph representing $T$. Let the leaves in $T^-$ be $v_1, \ldots, v_\ell$, where $\ell=\ell(T)$,
and let $u_i$ be the unique neighbor of $v_i$ in $T^-$. Note that by the assumption on $T$, no vertex $u_i$ can be identical to any of the leaves $v_j$.
Define $T_{v_i}$ to be the subtree of $T$ rooted at $v_i$, containing $v_i$ and its leaf neighbors in $T$.

We construct an injection $f$ by uniquely associating with each subtree $T_{v_i}$ a hyperarc of $H$ having a multi-headed body.
Let $\Heads\edge{u}{v}$ denote the set of heads of hyperarc in $H$ having body $(u,v)$.

\medskip

Step 1: Consider all subtrees $T_{v_i}$ containing a multi-headed body $\edge{v_i}{w}$ of $H$. If the set $\Heads\edge{v_i}{w}$ contains a vertex
$x\not\in \cup_{j\neq i} T_{v_j}$ then define 
\begin{equation}\label{eq:First}
f(i): =(\cl{v_i}{w}{x}).
\end{equation}

Assume that $\Heads\edge{v_i}{w}\subseteq \cup_{j\neq i} T_{v_j}$. The set $\{v_i,w\}\cup \Heads\edge{v_i}{w}$ must contain an edge other than $\edge{v_i}{w}$,
 otherwise forward chaining cannot mark any further vertices. Thus there must be at least one other subtree $T_{v_j}$ for $j\neq i$ containing two heads of $\edge{v_i}{w}$. Let the two heads be $x_1$ and $x_2$ and define 
 \begin{equation}\label{eq:Second}
 f(i): =(\cl{v_i}{w}{x_1}), \quad\quad\quad
 f(j): =(\cl{v_i}{w}{x_2}).
 \end{equation}
 
 The hyperarcs selected in this step are all distinct.

 \medskip

Step 2: Consider all subtrees $T_{v_i}$ for which $f$ is undefined after Step 1 such that $\edge{u_i}{v_i}$ is multi-headed in $H$.
Again, if the set $\Heads\edge{v_i}{w}$ contains a vertex  $x\not\in \cup_{j\neq i} T_{v_j}$ then define 
 \begin{equation}\label{eq:Third}
f(i): =(\cl{u_i}{v_i}{x}).
\end{equation}
Note that the body $(u_i, v_i)$ has not been used so far and so $f$ remains injective.

Assume that $\Heads\edge{u_i}{v_i}\subseteq \cup_{j\neq i} T_{v_j}$. The body $(u_i, v_i)$ has to reach the leaves of $T_{v_i}$, hence there must be hyperarcs
with body outside $T_{v_i}$ and having a leaf of $T_{v_i}$ as their head. The body of such a hyperarc must be multi-headed as otherwise forward chaining
could not continue. Let $\cl{x}{y}{w}$ be any such hyperarc and put 
 \begin{equation}\label{eq:Fourth}
f(i): =(\cl{x}{y}{w}).
\end{equation}
The partial mapping is still injective, as hyperarcs of type (\ref{eq:First}) and (\ref{eq:Third}) have heads outside the subtrees, and hyperarcs of type (\ref{eq:Second}) have heads
in subtrees for which $f$ is defined in Step 1.

\medskip

Step 3: Consider all subtrees $T_{v_i}$ for which $f$ is undefined after Steps 1 and 2. For these subtrees it holds that
edge $\edge{u_i}{v_i}$ and edges $\edge{v_i}{w}$, where $w$ is a head in $T_{v_i}$, are all single-headed in $H$.
Also, $\edge{u_i}{v_i}$ must have its only head outside $T_{v_i}$, else no edge in the subtree $T_{v_i}$ can reach any vertex outside $T_{v_i}$.
Indeed, if there is a hyperarc $\cl{v_i}{w}{x}$, where $w$ is a leaf of $T_{v_i}$ and $x \not\in T_{v_i}$, then the body $(v_i, w)$ can either reach $u_i$ only
(if $x = u_i$) or no vertex other than its head.

Repeating the argument in the second half of Step 2, there must be a multi-headed hyperarc
$\cl{x}{y}{w}$ with body outside $T_{v_i}$ and having a leaf of $T_{v_i}$ as its head.
Then we define 
\begin{equation}\label{eq:Fifth}
f(i): =\cl{x}{y}{w}.
\end{equation}
The partial mapping is still injective, as hyperarcs of type~(\ref{eq:First}) and~(\ref{eq:Third}) have heads outside the subtrees, and hyperarcs of type~(\ref{eq:Second}) and~(\ref{eq:Fourth}) have heads
in subtrees for which $f$ is defined in Steps 1 and 2.

\medskip

Thus $H$ contains at least $\ell$ hyperarcs having multi-headed bodies. Let these hyperarcs have $k$ different bodies.
Then the excess of $H$ is at least $k$, as every multi-headed body contributes at least one to the excess,
and at least $\ell - k$, as the $k$ edges altogether contribute at least $\ell - k$ to the excess. Thus the excess is
at least $\max(k, \ell - k) \ge \lceil \ell / 2 \rceil$.
\end{proof}

Theorem~\ref{thm:LowerBoundTrees} can also be formulated for general graphs containing outer induced subgraphs that are trees (i.e., graphs that contain some non-leaf vertices that become leaves after all their neighbors of a degree-1 are deleted). The proof is essentially the same.

\begin{corollary}
\label{thm:LowerBoundGeneral}
Let $G^-$ be the graph obtained from $G$ by removing all degree-1 vertices, and define $\ell(G)$ to be the number of degree-1 vertices in $G^-$. If $\ell(G)>1$ then
\[\h(G)\geq |E(G)|+\left\lceil{\ell(G)}/{2}\right\rceil.\]
\end{corollary}

\section{Special classes of trees}\label{sect:specialTrees}
\subsection{Spiders}\label{subsect:Spiders}

A \emph{spider} is a tree with at most one vertex of degree at greater than~2.
Let $u$ be the vertex with degree at least 3. For convenience, we call the unique path from a leaf of a spider to $u$, including the endpoints, a \emph{leg}. The upper bound from Theorem~\ref{thm:HydraUpperBound} is sharp for spider trees (even if we require the subgraph $G'$ to be connected) and so is the lower bound from Theorem~\ref{thm:LowerBoundTrees}.

\begin{corollary}\label{cor:SpiderTreeHydraNumber}
If $T$ is a spider tree with $\ell$ legs of length at least two, then
\[ \h(T)=|E(T)| + \left\lceil \ell / 2 \right\rceil. \]
\end{corollary}

\begin{proof}
The upper bound is a corollary of Theorem~\ref{thm:HydraUpperBound} %(\ref{ii})
as $L(T)$ can be covered with $\left\lceil \ell/2\right\rceil$ paths obtained from pairing up the legs of the tree and taking their line graphs.
 The lower bound is a corollary of Theorem~\ref{thm:LowerBoundTrees}, since $\ell=\ell(T)$.
\end{proof}

Let $T_k$ to be the spider tree with $k$ legs of length exactly 2.

\begin{theorem}\label{thm:maxHydraNumberOfTrees}
The maximum hydra number of a tree $T$ is $\left\lfloor(5|V(T)|-3)/4\right\rfloor$, and this is attained by the trees $T_k$.
\end{theorem}

\begin{proof}

The maximum interval number of tree is $\left\lfloor(5|V(G)|-3)/4\right\rfloor$ (shown in~\cite{AndreaeAi1989}). By Corollary~\ref{cor:TotalIntervalNumber} $\h(T)\leq\tau(T)$ and the same upper bound is established for the maximum interval number of a tree.

On the other hand recall that $T_k$ is a spider tree with $k$ legs of length 2, $|V(T_k)|=2k+1$, $|E(T_k)|=2k$ and hydra number
{$\h(T_k)=|E(T_k)|+\left\lceil{k}/{2}\right\rceil = \left\lfloor(5|V(T_k)|-3)/4\right\rfloor.$}
\end{proof}

\begin{remark}Recall that Proposition~\ref{prop:SpanningSubgraph} implies that any one edge added to $T_k$ can increase the hydra number by at most one, and cannot increase the \emph{excess} of the hydra number, the quantity ($\h(G)-|E(G)|$). Every connected graph contains a spanning tree, thus the maximum value of the excess $\h(G)-|E(G)|$ for any connected graph $G$ is attained by the tree $T_k$.\end{remark}

In the discussion of the hydra number we considered single-headed graphs, which imposes a restriction on the \emph{outdegrees} of hyperarcs.
Now we turn to the discussion of \emph{indegrees} and consider the largest in-degree of a vertex in an optimal $H$. 
The following results show that the trees $T_k$ introduced above have interesting properties with respect to the indegree of their
central vertex: that vertex has large indegree in every optimal representation, and bounding its indegree results in a blow-up of the
size of its representations.

\begin{theorem}\label{thm:SpiderTreeInDegree}
The central vertex $u$ is the head of at least $\lfloor k/2 \rfloor$ hyperarcs in every optimal hydra $H$ representing $T_k$.
\end{theorem}

\begin{proof}
Let $u$ be the central vertex of $T_k$, $v_i$ be one of the $k=\ell(T)$ neighbors of $u$, and $w_i$ be the leaf neighbor of $v_i$. Consider an optimal hypergraph representing $T_k$ with $|E(T_k)|+\lceil k/2 \rceil$ hyperarcs. Either $\edge{v_i}{w_i}$ is single-headed in $H$, in which case $u$ must be its head, or it is multi-headed in $H$. By the optimality of $H$ exactly at most $\lceil k/2\rceil$ of these edges is multi-headed, and thus at least $\lfloor k/2\rfloor$ of them must be single-headed.
\end{proof}

The trade-off between the indegree of the center and the size of representations can be formulated as follows.

\begin{theorem}\label{thm:SpiderTreeRestrictedInDegree}
For any $k$ and any $d$, such that $1 \leq d \leq \lfloor k/2 \rfloor$, if $H$ is a hypergraph representing $T_k$, such that the center $u$ is a head of at most $d$ hyperarcs, then the number of hyperarcs in $H$ is at least $|E(T_k)|+k-d$.
\end{theorem}

\begin{proof}
As before, each single-headed body $\edge{v_i}{w_i}$ must have $u$ as its head, thus at most $d$ of these bodies are single-headed. Therefore at least $k-d$ of them are multi-headed and so the excess is at least $k-d$.
\end{proof}

\subsection{Complete Binary Trees}\label{subsubsect:BinaryTrees}

In this section we obtain upper and lower bounds for $\h(G)$ when $G$ is a complete binary tree. A \emph{complete binary tree} of depth $d$, denoted $B_d$, is a tree with $d+1$ levels, where every node on levels 1 through $d$ has exactly 2 children. $B_d$ has $2^{d+1}-1$ vertices and $2^{d+1}-2$ edges.

\begin{theorem}\label{thm:BinaryTreeBounds}
For $d \ge 3$ it holds that
\begin{equation}
\frac{9}{8}\left|E\left(B_{d}\right)\right|
\leq \h\left(B_{d}\right)\leq \frac{17}{15}\left|E\left(B_{d}\right)\right|+1.
\end{equation}
\end{theorem}

\begin{proof}
The lower bound follows from Theorem~\ref{thm:LowerBoundTrees}, since there are $2^{d-1}$ leaves in $(B_{d})^{-}$ (all vertices in the $(d-1)$ level of the tree).

\smallskip
For the upper bound, recall that $\p(T)$ is the minimum path cover number of $L(T')$ over all subgraphs $T'$ with no isolated vertices. We show that there is a spanning subgraph $B'$ of $B_{d}$ with no isolated vertices with a path cover containing at most $\left\lceil(2/15)|E(B_{d})|\right\rceil$ paths. Specifically, we delete every fourth level of edges in the binary tree, leaving a forest of $B_3$'s. This corresponds to deleting every fourth level of vertices in $L(B_{d})$, to obtain $L(B')$, and gives 2 paths in $L(B_d)$ for every 15 edges in the tree. (See Figure~\ref{fig:B4k+3}.)
%
%%%%%%%%%%%%%%%%%%%%%%%%%%%%%%%%%%%%%%%%%%%%%%%%%%%%%%%%%
%%%%%%%%%%%%%%%%%%%%%%% begin of figure %%%%%%%%%%%%%%%%%%%%%%%%%%
%%%%%%%%%%%%%%%%%%%%%%%%%%%%%%% %%%%%%%%%%%%%%%%%%%%%%%%%%
\begin{figure}[t]
\centering
\begin{pspicture}(25,-40)(340,200)
\psset{linewidth=.9, labelsep=1}
\cnode*(85,160){2.3pt}{v0}%\nput{180}{v0}{$v_0$}
\cnode*(285,160){2.3pt}{v1}%\nput{0}{v1}{$v_1$}
\ncline{v0}{v1}
\cnode*(40,105){2.3pt}{v2}%\nput{180}{v2}{$v_2$}
\cnode*(140,105){2.3pt}{v3}%\nput{0}{v3}{$v_3$}
\ncline{v2}{v3}
\cnode*(240,105){2.3pt}{v4}%\nput{180}{v4}{$v_4$}
\cnode*(340,105){2.3pt}{v5}%\nput{0}{v5}{$v_5$}
\ncline{v4}{v5}
\ncline[linecolor=red]{v0}{v2}\ncline[linecolor=red]{v0}{v3}
\ncline[linecolor=red]{v1}{v4}\ncline[linecolor=red]{v1}{v5}
%%%%%%%%%%%%%%%%%%%%%%%%%%%%%%%%%%%%%%%%%%%%%%%%%%%%%%
%%%%%%%%%%%%%%%%%%%%%%%%%%%%%%%%%%%%%%%%%%%%%%%%%%%%%%
%%%%%%%%%%%%%%%%%%%%%%%%%%%%%%%%%%%%%%%%%%%%%%%%%%%%%%
%%%%%%%%%%%%%%%%%%%%%%%% v6 %%%%%%%%%%%%%%%%%%%%%%%%%%
\cnode*(15,70){2.3pt}{v6}%\nput{180}{v6}{$v_6$}
\cnode(0,45){2.3pt}{v14}%\nput{270}{v14}{$v_{14}$}
\cnode(30,45){2.3pt}{v15}%\nput{270}{v15}{$v_{15}$}
\ncline[linestyle=dotted]{v14}{v15}\ncline[linestyle=dotted]{v6}{v14}\ncline[linestyle=dotted]{v6}{v15}
\cnode*(-5,33){2.3pt}{v30}%\nput{270}{v14}{$v_{14}$}
\cnode*(5,33){2.3pt}{v31}%\nput{270}{v14}{$v_{14}$}
\cnode*(25,33){2.3pt}{v32}%\nput{270}{v14}{$v_{14}$}
\cnode*(35,33){2.3pt}{v33}%\nput{270}{v15}{$v_{15}$}
\ncline{v30}{v31}\ncline[linestyle=dotted]{v14}{v30}\ncline[linestyle=dotted]{v14}{v31}
\pspolygon[fillcolor=gray](-9,0)(9,0)(9,37)(-9,37)
\rput(-5, -15){\psshadowbox{\psrotateleft{$L(B_{d-4})$}}}
\ncline{v32}{v33}\ncline[linestyle=dotted]{v15}{v32}\ncline[linestyle=dotted]{v15}{v33}
\pspolygon[fillcolor=gray](21,0)(39,0)(39,37)(21,37)
\rput(25, -15){\psshadowbox{\psrotateleft{$L(B_{d-4})$}}}
%%%%%%%%%%%%%%%%%%%%%%%% v7 %%%%%%%%%%%%%%%%%%%%%%%%%%
\cnode*(65,70){2.3pt}{v7}%\nput{0}{v7}{$v_7$}
\cnode(50,45){2.3pt}{v16}%\nput{270}{v16}{$v_{16}$}
\cnode(80,45){2.3pt}{v17}%\nput{270}{v17}{$v_{17}$}
\ncline[linestyle=dotted]{v16}{v17}\ncline[linestyle=dotted]{v7}{v16}\ncline[linestyle=dotted]{v7}{v17}
\ncline[linecolor=red]{v2}{v6}\ncline{v2}{v7}\ncline[linecolor=red]{v6}{v7}
\cnode*(45,33){2.3pt}{v34}%\nput{270}{v14}{$v_{14}$}
\cnode*(55,33){2.3pt}{v35}%\nput{270}{v14}{$v_{14}$}
\cnode*(75,33){2.3pt}{v36}%\nput{270}{v14}{$v_{14}$}
\cnode*(85,33){2.3pt}{v37}%\nput{270}{v15}{$v_{15}$}
\ncline{v34}{v35}\ncline[linestyle=dotted]{v16}{v34}\ncline[linestyle=dotted]{v16}{v35}
\pspolygon[fillcolor=gray](41,0)(59,0)(59,37)(41,37)
\rput(45, -15){\psshadowbox{\psrotateleft{$L(B_{d-4})$}}}
\ncline{v36}{v37}\ncline[linestyle=dotted]{v17}{v36}\ncline[linestyle=dotted]{v17}{v37}
\pspolygon[fillcolor=gray](71,0)(89,0)(89,37)(71,37)
\rput(75, -15){\psshadowbox{\psrotateleft{$L(B_{d-4})$}}}
%%%%%%%%%%%%%%%%%%%%%%%%%%%%%%%%%%%%%%%%%%%%%%%%%%%%%%
%%%%%%%%%%%%%%%%%%%%%%%%%%%%%%%%%%%%%%%%%%%%%%%%%%%%%%
%%%%%%%%%%%%%%%%%%%%%%%%%%%%%%%%%%%%%%%%%%%%%%%%%%%%%%
%%%%%%%%%%%%%%%%%%%%%%%% v8 %%%%%%%%%%%%%%%%%%%%%%%%%%
\cnode*(115,70){2.3pt}{v8}
\cnode(100,45){2.3pt}{v18}
\cnode(130,45){2.3pt}{v19}
\ncline[linestyle=dotted]{v18}{v19}\ncline[linestyle=dotted]{v8}{v18}\ncline[linestyle=dotted]{v8}{v19}
\cnode*(95,33){2.3pt}{v38}
\cnode*(105,33){2.3pt}{v39}
\cnode*(125,33){2.3pt}{v40}
\cnode*(135,33){2.3pt}{v41}
\ncline{v38}{v39}\ncline[linestyle=dotted]{v18}{v38}\ncline[linestyle=dotted]{v18}{v39}
\pspolygon[fillcolor=gray](91,0)(109,0)(109,37)(91,37)
\rput(95, -15){\psshadowbox{\psrotateleft{$L(B_{d-4})$}}}
\ncline{v40}{v41}\ncline[linestyle=dotted]{v19}{v40}\ncline[linestyle=dotted]{v19}{v41}
\pspolygon[fillcolor=gray](121,0)(139,0)(139,37)(121,37)
\rput(125, -15){\psshadowbox{\psrotateleft{$L(B_{d-4})$}}}
%%%%%%%%%%%%%%%%%%%%%%%% v9 %%%%%%%%%%%%%%%%%%%%%%%%%%
\cnode*(165,70){2.3pt}{v9}
\cnode(150,45){2.3pt}{v20}
\cnode(180,45){2.3pt}{v21}
\ncline[linestyle=dotted]{v20}{v21}\ncline[linestyle=dotted]{v9}{v20}\ncline[linestyle=dotted]{v9}{v21}
\cnode*(145,33){2.3pt}{v42}
\cnode*(155,33){2.3pt}{v43}
\cnode*(175,33){2.3pt}{v44}
\cnode*(185,33){2.3pt}{v45}
\ncline{v42}{v43}\ncline[linestyle=dotted]{v20}{v42}\ncline[linestyle=dotted]{v20}{v43}
\pspolygon[fillcolor=gray](141,0)(159,0)(159,37)(141,37)
\rput(145, -15){\psshadowbox{\psrotateleft{$L(B_{d-4})$}}}
\ncline{v44}{v45}\ncline[linestyle=dotted]{v21}{v44}\ncline[linestyle=dotted]{v21}{v45}
\pspolygon[fillcolor=gray](171,0)(189,0)(189,37)(171,37)
\rput(175, -15){\psshadowbox{\psrotateleft{$L(B_{d-4})$}}}
\ncline{v3}{v8}\ncline[linecolor=red]{v3}{v9}\ncline[linecolor=red]{v8}{v9}
%%%%%%%%%%%%%%%%%%%%%%%%%%%%%%%%%%%%%%%%%%%%%%%%%%%%%%
%%%%%%%%%%%%%%%%%%%%%%%%%%%%%%%%%%%%%%%%%%%%%%%%%%%%%%
%%%%%%%%%%%%%%%%%%%%%%%%%%%%%%%%%%%%%%%%%%%%%%%%%%%%%%
%%%%%%%%%%%%%%%%%%%%%%%% v10 %%%%%%%%%%%%%%%%%%%%%%%%%
\cnode*(215,70){2.3pt}{v10}
\cnode(200,45){2.3pt}{v22}
\cnode(230,45){2.3pt}{v23}
\ncline[linestyle=dotted]{v22}{v23}\ncline[linestyle=dotted]{v10}{v22}\ncline[linestyle=dotted]{v10}{v23}
\cnode*(195,33){2.3pt}{v46}
\cnode*(205,33){2.3pt}{v47}
\cnode*(225,33){2.3pt}{v48}
\cnode*(235,33){2.3pt}{v49}
\ncline{v46}{v47}\ncline[linestyle=dotted]{v22}{v46}\ncline[linestyle=dotted]{v22}{v47}
\pspolygon[fillcolor=gray](191,0)(209,0)(209,37)(191,37)
\rput(195, -15){\psshadowbox{\psrotateleft{$L(B_{d-4})$}}}
\ncline{v48}{v49}\ncline[linestyle=dotted]{v23}{v48}\ncline[linestyle=dotted]{v23}{v49}
\pspolygon[fillcolor=gray](221,0)(239,0)(239,37)(221,37)
\rput(225, -15){\psshadowbox{\psrotateleft{$L(B_{d-4})$}}}
%%%%%%%%%%%%%%%%%%%%%%%% v11 %%%%%%%%%%%%%%%%%%%%%%%%%%
\cnode*(265,70){2.3pt}{v11}
\cnode(250,45){2.3pt}{v24}
\cnode(280,45){2.3pt}{v25}
\ncline[linestyle=dotted]{v24}{v25}\ncline[linestyle=dotted]{v11}{v24}\ncline[linestyle=dotted]{v11}{v25}
\cnode*(245,33){2.3pt}{v50}
\cnode*(255,33){2.3pt}{v51}
\cnode*(275,33){2.3pt}{v52}
\cnode*(285,33){2.3pt}{v53}
\ncline{v50}{v51}\ncline[linestyle=dotted]{v24}{v50}\ncline[linestyle=dotted]{v24}{v51}
\pspolygon[fillcolor=gray](241,0)(259,0)(259,37)(241,37)
\rput(245, -15){\psshadowbox{\psrotateleft{$L(B_{d-4})$}}}
\ncline{v52}{v53}\ncline[linestyle=dotted]{v25}{v52}\ncline[linestyle=dotted]{v25}{v53}
\pspolygon[fillcolor=gray](271,0)(289,0)(289,37)(271,37)
\rput(275, -15){\psshadowbox{\psrotateleft{$L(B_{d-4})$}}}
\ncline[linecolor=red]{v4}{v10}\ncline{v4}{v11}\ncline[linecolor=red]{v10}{v11}
%%%%%%%%%%%%%%%%%%%%%%%% v12 %%%%%%%%%%%%%%%%%%%%%%%%%
%%%%%%%%%%%%%%%%%%%%%%%%%%%%%%%%%%%%%%%%%%%%%%%%%%%%%%
%%%%%%%%%%%%%%%%%%%%%%%%%%%%%%%%%%%%%%%%%%%%%%%%%%%%%%
\cnode*(315,70){2.3pt}{v12}
\cnode(300,45){2.3pt}{v26}
\cnode(330,45){2.3pt}{v27}
\ncline[linestyle=dotted]{v26}{v27}\ncline[linestyle=dotted]{v12}{v26}\ncline[linestyle=dotted]{v12}{v27}
\cnode*(295,33){2.3pt}{v54}
\cnode*(305,33){2.3pt}{v55}
\cnode*(325,33){2.3pt}{v56}
\cnode*(335,33){2.3pt}{v57}
\ncline{v54}{v55}\ncline[linestyle=dotted]{v26}{v54}\ncline[linestyle=dotted]{v26}{v55}
\pspolygon[fillcolor=gray](291,0)(309,0)(309,37)(291,37)
\rput(295, -15){\psshadowbox{\psrotateleft{$L(B_{d-4})$}}}
\ncline{v56}{v57}\ncline[linestyle=dotted]{v27}{v56}\ncline[linestyle=dotted]{v27}{v57}
\pspolygon[fillcolor=gray](321,0)(339,0)(339,37)(321,37)
\rput(325, -15){\psshadowbox{\psrotateleft{$L(B_{d-4})$}}}
%%%%%%%%%%%%%%%%%%%%%%%% v13 %%%%%%%%%%%%%%%%%%%%%%%%%%
\cnode*(365,70){2.3pt}{v13}
\cnode(350,45){2.3pt}{v28}
\cnode(380,45){2.3pt}{v29}
\ncline[linestyle=dotted]{v28}{v29}\ncline[linestyle=dotted]{v13}{v29}\ncline[linestyle=dotted]{v13}{v28}
\cnode*(345,33){2.3pt}{v58}
\cnode*(355,33){2.3pt}{v59}
\cnode*(375,33){2.3pt}{v60}
\cnode*(385,33){2.3pt}{v61}
\ncline{v58}{v59}\ncline[linestyle=dotted]{v28}{v58}\ncline[linestyle=dotted]{v28}{v59}
\pspolygon[fillcolor=gray](341,0)(359,0)(359,37)(341,37)
\rput(355, -15){\psshadowbox{\psrotateleft{$L(B_{d-4})$}}}
\ncline{v60}{v61}\ncline[linestyle=dotted]{v29}{v60}\ncline[linestyle=dotted]{v29}{v61}
\pspolygon[fillcolor=gray](371,0)(389,0)(389,37)(371,37)
\rput(375, -15){\psshadowbox{\psrotateleft{$L(B_{d-4})$}}}
\ncline{v5}{v12}\ncline[linecolor=red]{v5}{v13}\ncline[linecolor=red]{v12}{v13}
\end{pspicture}
\caption{Induction on the line graph $L(B_{d})$ from the proof of Theorem~\ref{thm:BinaryTreeBounds}.}
\label{fig:B4k+3}
\end{figure}
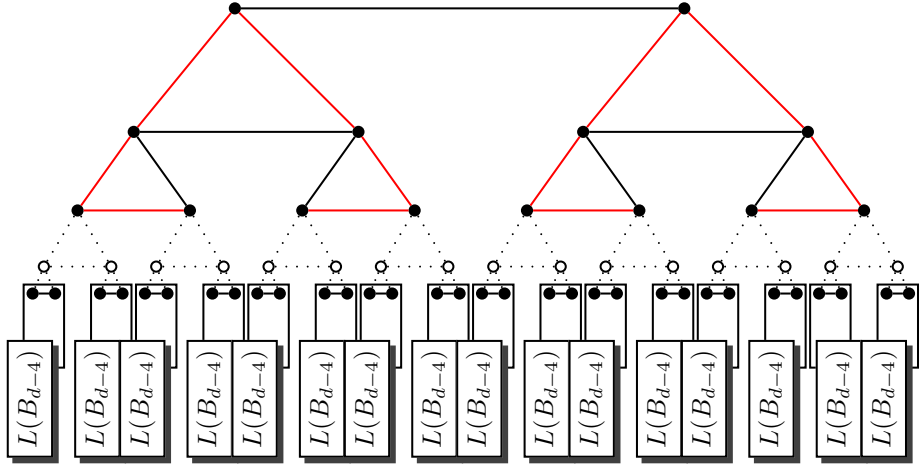
%%%%%%%%%%%%%%%%%%%%%%%%%%%%%%%%%%%%%%%%%%%%%%%%%%%%%%%%%
%%%%%%%%%%%%%%%%%%%%%%%% end of figure %%%%%%%%%%%%%%%%%%%%%%%%%%
%%%%%%%%%%%%%%%%%%%%%%%%%%%%%%% %%%%%%%%%%%%%%%%%%%%%%%%%%
\end{proof}

Recall that $\pcn(G)$ is the path cover number of graph $G$, and let {$\g(G)=\pcn(L(G))$}. The following theorem demonstrates the significance of considering non-connected subgraphs of $G$ when upper-bounding $\h(G)$ in Theorem~\ref{thm:HydraUpperBound}. It demonstrates that even for trees, there are examples where the quantity $\g(G)$ is significantly larger than the quantity $\p(G)$.

\begin{theorem}\label{thm:BinaryTreePaths}
For $d$ large,
 \[\g(B_d)-\p(B_d)=\Theta(|V(B_d)|).\]
\end{theorem}

\begin{proof}
Applying the algorithm of  \cite{KratzkeWe1996} the path cover number of $L(B_d)$ is $\g(B_d)=\left\lceil\left|E\left(B_{d}\right)\right|/7 \right\rceil$. This is because this algorithm calculates the trail cover number of a tree, which by~\cite[Theorems~3 \& 6]{Raychaudhuri1995} is equivalent to the minimum number of caterpillars required to cover the edges of a tree, and the path cover number of the line graph. One can construct an optimal caterpillar edge cover from the bottom of the tree, and from left to right, giving one caterpillar, and equivalently one path in the path cover of $L(B_d)$, for every 7 edges. \end{proof}

It would be interesting to close the gap between the bounds in Theorem~\ref{thm:BinaryTreeBounds}. The lower bound takes only the bottom part of the tree into consideration and thus it is not expected to be sharp. It is an open question whether the upper bound $|E(T)| + p(T)$ is sharp for trees in general.

%%%%%%%%%%%%%
\section{A related minimization problem}\label{sect:CompleteGraphs}

There are several ways in which one can generalize hydra numbers. One can consider hydra numbers for hypergraphs, where the bodies of the hyperarcs contain more than two vertices and so they form a hypergraph instead of a graph. Studying hydra numbers of hypergraphs is an open area of research. In this section we consider a second generalization of hydra numbers.

Given $n$ and a number $k$ ($2\leq k\leq n-1$), let $f(n, k)$ be the minimal number of hyperarcs in an $n$-vertex 3-uniform directed hypergraph $H$ such that the closure of every $k$-element subset of the vertices is the whole vertex set. The case $k = 2$ is just the hydra number of complete graphs and so $f(n,2)={n\choose 2}$.

We use Tur\'an's theorem from extremal graph theory (see, e.g. \cite{West2001}).
The \emph{Tur\'an graph} $T(n, k - 1)$ is formed by dividing $n$ vertices into
$k - 1$ parts as evenly as possible (i.e., into parts of size $\lfloor n/(k - 1) \rfloor$
and $\lceil n/(k - 1) \rceil$) and connecting two vertices iff they are in different
parts. The number of edges of $T(n, k - 1)$ is denoted by $t(n, k - 1)$. If $k - 1$
divides $n$ then
$$t(n, k - 1) = \left(1-\frac{1}{k-1}\right) \frac{n^2}{2}.$$
Tur\'an's theorem states that if an $n$-vertex graph contains no $k$-clique
then it has at most $t(n, k - 1)$ edges and the only extremal graph is
$T(n, k - 1)$. Switching to complements it follows that if an $n$-vertex graph has
no empty subgraph on $k$ vertices then it has at least ${n \choose 2} - t(n, k - 1)$
edges.

\begin{theorem}\label{thm:BeyondHydras} If $k \le (n/2) + 1$ then
\[ {n \choose 2} - t(n, k - 1) \le f(n, k) \le {n \choose 2} - t(n, k - 1) + (k - 1). \]
\end{theorem}

\begin{proof}
Suppose $H$ is a 3-uniform directed hypergraph with all $k$-tuples having the whole set as their closure. Then every $k$-element set $S$ of vertices must contain at least one body of a hyperarc in $H$, otherwise forward chaining started from $S$ cannot mark any vertices. Thus the undirected graph formed by the bodies in $H$ contains no empty subgraph on $k$ vertices, and the lower bound follows by Tur\'an's theorem.

For the upper bound we construct a directed hypergraph based on the complement of $T(n, k - 1)$ over the vertex set $\{x_1, \ldots, x_n\}$, consisting of $k - 1$ cliques of size differing by at most 1. Assume that each clique has size at least 3. In each clique do the following. Pick a Hamiltonian path, direct it, and introduce hyperarcs as in~(\ref{eq:cycfor}) (with the exception of the last edge closing the cycle).
For every other edge $(u,v)$, introduce a hyperarc $\cl{u}{v}{w}$ where $w$ is a vertex on the Hamiltonian path that is adjacent to $u$ or $v$. For each edge $e$ closing a Hamiltonian cycle, add \emph{two} hyperarcs with body $e$, and heads the endpoints of the first edge on the Hamiltonian path of the next clique (where `next' assumes an arbitrary cyclic ordering of the cliques). For cliques of size 2 the single edge in the clique plays the role of the unassigned edge and the construction is similar.
\end{proof}

\section{Acknowledgements}
This material is based upon work supported by the National Science Foundation under Grant No. CCF-0916708. The second author acknowledges funding by the People Programme (Marie Curie Actions) of the European UnionÕs Seventh Framework Programme FP7/2007-2013/ under REA grant agreement n$^{\mathrm{o}}$ [316808].

%%%%%%%%%%
% References and End of Paper
\bibliography{hydras}
\bibliographystyle{plain}

\end{document}